\documentclass{article}
\usepackage{amsmath,amsfonts,amssymb,amsthm}
\usepackage{amsmath,amsthm,amssymb,amsfonts,graphicx,amsthm,nicefrac,mathtools}
\usepackage{enumerate, algorithm,algpseudocode}
\usepackage{mathtools,dsfont}

\usepackage{pstricks-add}
\usepackage[round]{natbib}
\usepackage{graphicx}
\usepackage{paralist}
 
\allowdisplaybreaks

\newcommand{\ca}{{\cal A}}

\newcommand{\cx}{{\cal X}}

\newcommand{\bsu}{\boldsymbol{u}}
\newcommand{\bsx}{\boldsymbol{x}}
\newcommand{\bsy}{\boldsymbol{y}}
\newcommand{\bsz}{\boldsymbol{z}}

\newcommand{\natu}{\mathbb{N}}
\newcommand{\real}{\mathbb{R}}

\newcommand{\mrd}{\mathrm{d}}
\newcommand{\rd}{\mathrm{\, d}}

\newcommand{\vol}{{\mathbf{vol}}}

\newcommand{\wt}{\widetilde}

\renewcommand{\emptyset}{\varnothing}
\renewcommand{\ge}{\geqslant}
\renewcommand{\le}{\leqslant}

\newcommand{\dnorm}{\mathcal{N}}

\newcommand{\dustd}{\mathbf{U}} 

\newcommand{\e}{{\mathbb{E}}} 

  \newcommand{\BIT}{\begin{itemize}}
\newcommand{\EIT}{\end{itemize}}
\newcommand{\BNUM}{\begin{enumerate}}
\newcommand{\ENUM}{\end{enumerate}}

\theoremstyle{definition}

\theoremstyle{plain}
\newtheorem{theorem}{Theorem}
\newtheorem{corollary}{Corollary}

\newtheorem{lemma}{Lemma}
\theoremstyle{definition}

\newcommand{\px}{{\mathcal{X}}}

\newcommand{\bv}{{\mathrm{BV}}}
\newcommand{\vit}{{\mathrm{V}_{\mathrm{IT}}}}
\newcommand{\vhk}{{\mathrm{V}_{\mathrm{HK}}}}
\newcommand{\bvhk}{{\mathrm{BVHK}}}

\newcommand{\simiid} {{\stackrel{\mathrm{iid}}\sim}}

\newcommand{\bsmu}{{\boldsymbol{\mu}}}
\newcommand{\bsnu}{{\boldsymbol{\nu}}}
\newcommand{\bslam}{{\boldsymbol{\lambda}}}
\newcommand{\bsell}{{\boldsymbol{\ell}}}
\newcommand{\bsk}{{\boldsymbol{k}}}
\newcommand{\calkl}{{\mathrm{KL}}}
\newcommand{\calklt}{{\widetilde{\mathrm{KL}}}}
\newcommand{\bsg}{{\boldsymbol{g}}}
\newcommand{\bszero}{{\boldsymbol{0}}}
\newcommand{\bsone}{{\boldsymbol{1}}}

\date{December 2015}
\author{Kinjal Basu\\Stanford University \and Art B. Owen\\Stanford University}
\title{Transformations and Hardy-Krause variation}
\begin{document}
\maketitle

\begin{abstract}
Using a multivariable Faa di Bruno formula we give conditions
on transformations $\tau:[0,1]^m\to\cx$ where
$\cx$ is a closed and bounded subset of $\real^d$ such
that $f\circ\tau$ is of bounded variation in the sense of Hardy
and Krause for all $f\in C^d(\cx)$. We give similar conditions for
$f\circ\tau$ to be smooth enough for scrambled net sampling
to attain $O(n^{-3/2+\epsilon})$ accuracy. Some popular symmetric transformations
to the simplex and sphere are shown to satisfy neither condition.
Some other transformations due to Fang and Wang (1993)\nocite{fang1993number} satisfy
the first but not the second condition.  We provide transformations
for the simplex that makes $f\circ\tau$ smooth enough to fully
benefit from scrambled net sampling for all $f$ in a class of generalized
polynomials. We also find sufficient conditions for the Rosenblatt-Hlawka-M\"uck
transformation in $\real^2$ and for importance sampling
to be of bounded variation in the sense of Hardy and Krause.
\end{abstract}

\section{Introduction}

Quasi-Monte Carlo (QMC) sampling is usually applied to integration problems over
the domain $[0,1]^d$.  Other domains, such as triangles, disks, simplices,
spheres, balls, et cetera are also of importance  in applications.
Monte Carlo (MC) sampling over such a domain $\px\subset\real^d$ 
is commonly done by finding a
uniformity preserving  transformation $\tau:[0,1]^m\to\px$. 
Such transformations yield $\bsx=\tau(\bsu)\sim\dustd(\px)$
when $\bsu\sim\dustd[0,1]^m$ so that
\begin{align}\label{eq:tauok}
\frac1{\vol(\px)}
\int_{\px}f(\bsx)\rd\bsx = \int_{[0,1]^m}f(\tau(\bsu))\rd\bsu. 
\end{align}
Then we estimate $\mu = \int_{\cx} f(\bsx)\rd\bsx$ by
$(\vol(\cx)/n)\sum_{i=1}^nf(\tau(\bsu_i))$ for $\bsu_i\simiid\,\dustd[0,1]^d$.
We will often take $\vol(\cx)=1$ for simplicity.

A very standard approach to QMC sampling of such domains is to 
employ the same transformation $\tau$ as in MC, but
to replace independent random $\bsu_i$ by QMC or randomized QMC (RQMC) points.
The uniformity-preserving transformation $\tau$ satisfies
equation~\eqref{eq:tauok} when $m=d$ and  $\tau$ has 
a Jacobian determinant everywhere equal to $1/\vol(\px)=1$. 
It also holds when that Jacobian determinant is piecewise constant 
and equal to $\pm1$ at all $\bsu$. 
Equation~\eqref{eq:tauok} does not require $m=d$. For instance,
in Section~\ref{sec:3to2},  we study 
a logarithmic mapping from $[0,1]^3$ to a two-dimensional equilateral 
triangle which satisfies~\eqref{eq:tauok}. 

When the function $f\circ\tau$ is of bounded variation 
in the sense of Hardy and Krause (BVHK), then the Koksma-Hlawka
inequality applies and QMC can attain the convergence rate
$O(n^{-1+\epsilon})$. Under additional smoothness conditions
on $f\circ \tau$, certain RQMC methods (scrambled nets) 
have a root mean squared error (RMSE) of
$O(n^{-3/2+\epsilon})$.

The paper proceeds as follows.
Section~\ref{sec:smoothness} gives a sufficient condition
for a function $f$ on $[0,1]^d$ to be in BVHK and a stronger
condition for that function to be integrable with RMSE $O(n^{-3/2+\epsilon})$
by scrambled nets. Those conditions are expressed in terms of
certain partial derivatives of $f$.
Section~\ref{sec:composition} considers how to apply the
conditions from Section~\ref{sec:smoothness} to compositions $f\circ\tau$.
There are good sufficient conditions for compositions of single variable
functions to be in BVHK, but the multivariate setting is more complicated
as shown by a counterexample there. Then we specialize a multivariable Faa di Bruno
formula from \cite{cons:savi:1996} to the mixed partial derivatives required
for QMC.
Section~\ref{sec:conditions} gives sufficient conditions for $f\circ\tau$
to be in BVHK and also for  it to be smooth enough for 
scrambled nets to improve on the QMC rate.  There is also a discussion
of necessary conditions. We stipulate there that at least
the components of $\tau$ should themselves be in BVHK.
Section~\ref{sec:counter} considers two widely used 
transformations $\tau$ that are symmetric operations on $d$
input variables to yield uniform points in the $d-1$ dimensional
simplex and sphere respectively. Unfortunately some components of $\tau$ fail to be in BVHK
for these transformations.
Section~\ref{sec:fangwang} shows that some classic mappings
to the simplex, sphere and ball from \cite{fang1993number}
are in BVHK, although they are not smooth enough to benefit from RQMC.
Section~\ref{sec:nonunif} considers non-uniform transformations including
importance sampling,  Rosenblatt-Hlawka-M\"uck sequential inversion,
and a non-uniform transformation on the unit simplex that yields the
customary RQMC convergence rate for a class of functions including all
polynomials on the simplex.

While finishing up this paper we noticed that
\cite{camb:hofe:lemi:2015} have also applied the Faa di Bruno
formula in a QMC application, though they apply it to
a different set of problems. They use it to give sufficient
conditions for some integrands with respect to copulas
to be in BVHK.  They extend \cite{hlaw:muck:1972} for inverse
CDF sampling to some copulas with
mixed partial derivatives that are singular on the boundaries of the unit cube.
They closely study the Marshall-Olkin algorithm which generates
points from a $d$ dimensional Archimedean copula from a point in $[0,1]^{d+1}$
and give conditions for quadrature errors to be bounded by 
a multiple of the $d+1$-dimensional discrepancy and weaker conditions
for a bound $\log(n)$ times as large as that.

\section{Smoothness conditions}\label{sec:smoothness}

Quasi-Monte Carlo sampling attains an error
rate of $O(n^{-1}(\log n)^{d-1})$, if the function $f\in\bvhk$.
Here we give a simply checked sufficient condition for $f\in\bvhk$.
We use $\vhk(f)$ for the total variation of $f$ in the sense
of Hardy and Krause and $\vit(f)$ for the total variation of $f$
in the sense of Vitali.

Let $1{:}d=\{1,2,\dots,d\}$.  For a set $u\subseteq1{:}d$,
let $|u|$ denote the cardinality of $u$ and $-u=1{:}d \setminus u$
its complement.
Let $\partial^u f$ denote the partial derivative of $f$
taken once with respect to each variable $j\in u$.
By convention $\partial^\emptyset f=f$.
For $\bsx\in[0,1]$ and $u\subseteq1{:}d$
let $\bsx_u{:}\bsone_{-u}$ be the point $\bsy\in[0,1]^d$
with $y_j=x_j$ for $j\in u$ and $y_j=1$ for $j\not\in u$.

If the mixed partial derivative $\partial^{1:d}f$ exists then
\begin{align}
\vit(f)& \le\int_{[0,1]^{d}} |\partial^{1{:}d} f(\bsx)|\rd\bsx,\quad\text{and}\label{eq:vitbound}\\
\vhk(f)& \le\sum_{u\ne\emptyset}\int_{[0,1]^{|u|}} |\partial^u f(\bsx_{u}{:}\bsone_{-u})|\rd\bsx_u. \label{eq:hkbound}
\end{align}
These and related results are presented in~\cite{variation}.
\cite{frec:1910} shows that
the Vitali bound~\eqref{eq:vitbound} becomes an equality if $\partial^{1{:}d}f$ is 
continuous on $[0,1]^d$.
The Hardy-Krause variation is a sum of Vitali variations for which
\eqref{eq:hkbound} arises by applying~\eqref{eq:vitbound} term by term.

For scrambled nets, a kind of RQMC, to attain a root mean
squared error of order $O(n^{-3/2}(\log n)^{(d-1)/2})$ the function $f$ must
be smooth in the following sense:
\begin{align}\label{eq:snetcond}
\Vert \partial^u f\Vert^2_2\equiv
\int (\partial^uf(\bsx))^2\rd\bsx <\infty,\quad\forall u\subseteq1{:}d.
\end{align}
For a description of digital nets including scramblings of them
see~\cite{dick:pill:2010}.
Two scramblings with RMSE of $O(n^{-3/2+\epsilon})$ are the nested uniform
scramble in \cite{rtms} and the nested linear scramble
of \cite{mato:1998}. Geometric nets and scrambled geometric nets have been introduced in \cite{basu2015scrambled}
for sampling uniformly on  $\cx^s$ where $\cx$ is a closed and bounded subset of $\real^d$.
Scrambled geometric nets attain an RMSE of $O(n^{-1/2 - 1/d} (\log n)^{(s - 1)/2})$ for certain smooth functions defined on $\cx^s$.
The construction of scrambled geometric nets is based on the recursive partitions used
by \cite{basu:owen:2015} to sample the triangle.

We will study transformations by considering which combinations
of $f$ and $\tau$ give $\vhk(f\circ\tau)<\infty$.  For such combinations
plain QMC will be asymptotically better than geometric nets when $s = 1$ and $d\ge3$.
Similarly, if $\partial^u (f\circ\tau)\in L^2$ for all $u\subseteq1{:}d$, then scrambled nets
are asymptotically better than geometric nets for $d\ge2$.

Higher order digital nets \citep{dick:mcqmc:2009} achieve even
better rates of convergence than plain (R)QMC does, 
but they require even stronger smoothness
conditions. Their randomized versions \citep{dick:2011} further increase
accuracy (in root mean square) under yet stronger smoothness
conditions.

\section{Function composition}\label{sec:composition}

We would like a condition under which
the composition $f\circ\tau:[0,1]^m\to\real^d\to\real$
is in BVHK.  
For the case $d=m=1$, BVHK for $f\circ\tau$ reduces to 
ordinary BV.
\cite{josephy1981composing}  
gives a very complete
characterization of when compositions of one dimensional
functions are in BV.

Let $f$ and $\tau$ be functions of bounded variation 
from $[0,1]$ to $[0,1]$. 
Theorem 4 of \cite{josephy1981composing}  
shows that $f\circ \tau\in\bv$ holds
for all $\tau\in\bv$ if and only if $f$ is Lipschitz. 
The statement on $\tau$ is a bit more complicated.
His Theorem 3 shows that $f\circ \tau\in \bv$
for all $f\in\bv$ if and only if $\tau$ belongs to a 
special subset of $\bv$, in which pre-images 
of intervals are unions of a finite set of intervals. 

\subsection{A counter-example}
No such comprehensive characterization is available for
bounded variation in the sense of Hardy and Krause in higher dimensions.
Here we present functions $f:\real^2\to\real$ and $\tau:\real^2\to\real^2$
such that $f$ is Lipschitz and $\tau\in\bvhk$ but $f\circ\tau\not\in\bvhk$.
We take $\tau$ to be the identity map on $[0,1]^2$ for which both
components are in BVHK.
Then we construct a Lipschitz function $f:[0,1]^2\to\real$ with
$f\circ\tau=f\not\in\bvhk$.

We define the function $f$ in a recursive way using a Sierpinsky gasket type splitting of the unit square.
Let  $A$ be the square $(x_1,x_1+\ell)\times (x_2,x_2+\ell)\subset[0,1]^2$
for some $\ell>0$.
Then for $\bsx'\in[0,1]^2$, define the pyramid function 
$$f_A(\bsx')=\max\Bigl(0,\frac\ell2-\max_{j=1,2}| x'_j-(x_j+\ell/2)|\Bigr).$$
This function is $0$ for $\bsx'\not\in A$ and inside $A$ it defines the upper
surface of a square based pyramid of height $\ell/2$ centered the midpoint of $A$.
For an illustration, see the lower right hand corner of the second
panel in Figure~\ref{fig:recursive}. For any $A$ the function $f_A$ is Lipschitz
continuous with Lipschitz constant $1$.

We construct $f$ as follows.  
First we split $[0,1]^2$ into four congruent
sub-squares as shown in the left panel of Figure~\ref{fig:recursive}.
Then we select one of those sub-squares, say $A_4$ and
initially set $f=f_{A_4}$.
Next, we partition each of the remaining three sub-squares $A_1,\dots,A_3$
into four congruent sub-sub-squares $A_{ij}$ for $i=1,\dots,3$ and $j=1,\dots,4$.
Then we add $f_{A_{1,4}}+f_{A_{2,4}} +f_{A_{3,4}}$ to $f$.
This construction is carried out recursively 
summing $3^{k}$ pyramidal functions at level $k=0,1,2,\dots$
over  squares of side $2^{-k-1}$,
as depicted in the right panel of Figure~\ref{fig:recursive}.

\begin{lemma}
The function $f$ described above has Lipschitz constant one 
and has infinite Vitali variation and hence infinite variation
in the sense of Hardy and Krause.
\end{lemma}
\begin{proof}
Let $\bsx, \bsy \in [0,1]^2$. Consider the function $g(t) = f(\bsx + t(\bsy-\bsx))$
on $0\le t\le 1$.  This function is continuous and piece-wise linear with
absolute slope at most $1$.  Thus $|f(\bsx)-f(\bsy)|\le\Vert\bsx-\bsy\Vert$
and so $f$ is Lipschitz with constant~$1$.

Now we turn to variation using definitions and results from~\cite{variation}.
The Vitali variation of $f_A$ 
equals the Vitali variation of $f_A$ over the square $A$. 
By considering a $3\times 3$ grid covering the edges
and center of $A$ we find that $\vit(f_A)\ge 2\ell$.
In fact, $\vit(f_A)=2\ell$ but we only need the lower bound.

The Vitali variation of $f$ is the sum of its Vitali variations
over a square subpartition.  As a result, $\vit(f)$ is the sum
of $\vit(f_A)$ for all the sets $A$ in our recursive construction.
For $k=0,1,2,\dots$ there are $3^k$ terms $f_A$ with $\ell = 2^{-k-1}$.
As a result the Vitali variation of $f$ is
at least $\sum_{k=0}^\infty 3^k2^{-k}=\infty$.
\end{proof}

\begin{figure}[t]
\label{fig:recursive}
\begin{center}
\includegraphics[width = \linewidth]{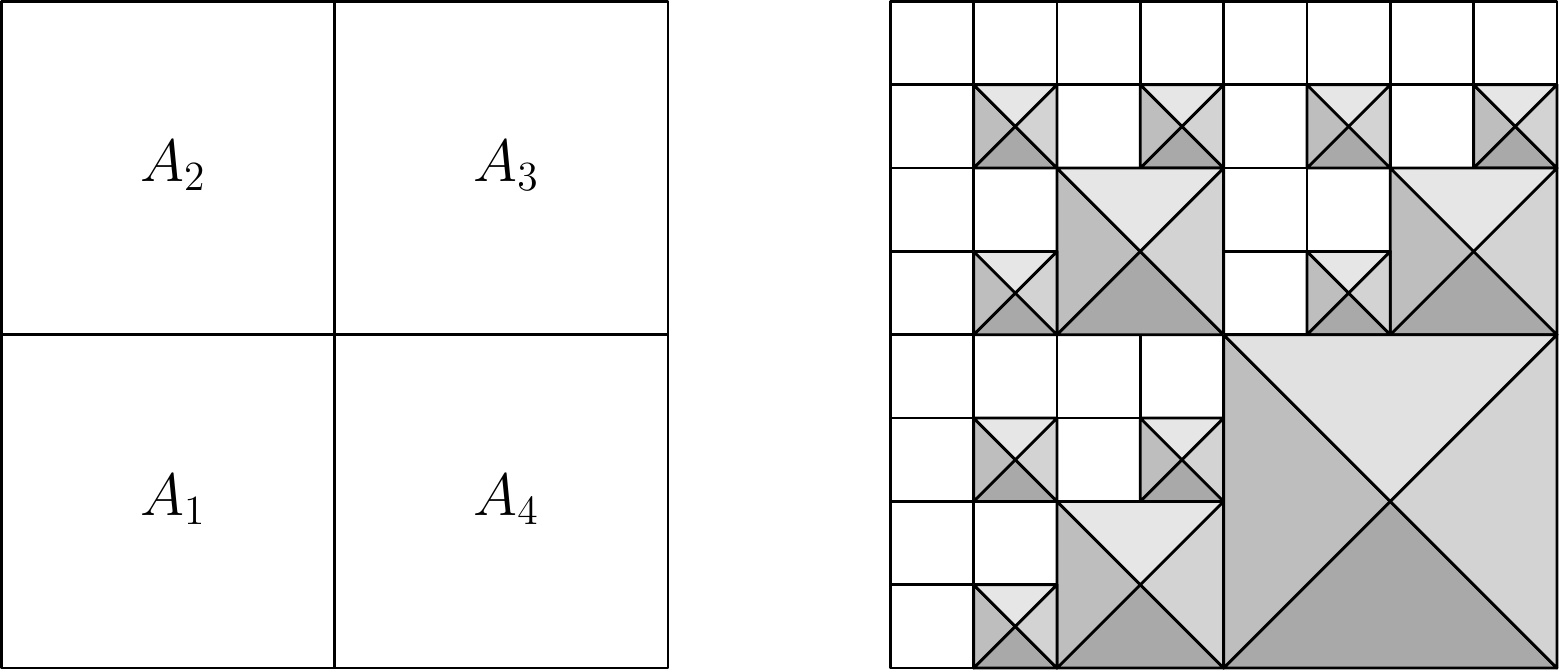}
\caption{The plot on the left shows the square partition $\mathcal{P}$ which is repeated in a recursive manner. The right figure shows the function as a 2-dimensional projection for $k = 3$. Each such pyramidal structure has a height of half the length of its base square.}
\end{center}
\end{figure}

This counterexample applies to any $d\ge2$ simply by constructing
a function that equals the above constructed function $f$ applied to
two of its input variables. Such a function has infinite Hardy-Krause
variation arising from the Vitali variation in those two variables.
As a result, even if $f$ is in Lipschitz and $\tau$ is BVHK along with every component, 
we might still have $f\circ\tau\not\in\bvhk$.

\subsection{Faa di Bruno formulas}

We will study variation via a mixed partial derivative
of the composition of the integrand on $\px$ with
a transformation from the unit cube to $\px$.
We need partial derivatives of order up to the
dimension of the unit cube. High order derivatives
of a composition become awkward even in the case with $d=m=1$,
which was solved by \cite{faad:1855}.
We will use a multivariable Faa di Bruno formula from
\cite{cons:savi:1996}.

To remain consistent with the notation in
\cite{cons:savi:1996} we will consider functions
$h = f(g(\cdot))$ here.  After obtaining the formulas
we need, we will revert to $f(\tau(\cdot))$ that
is more suitable for the MC and QMC context.

To illustrate Faa di Bruno suppose first that $d=m=1$.
Then let $g$ be defined on an open set containing $x_0$
and have derivatives up to order $n$ at $x_0$. 
Let $f$ be defined on an open set containing $y_0=g(x_0)$
and have derivatives up to order $n$ at $y_0$. 
For $0\le k\le n$ define derivatives 
$f_k=\mrd^kf(y_0)/\mrd y^k$, $g_k=\mrd^kg(x_0)/\mrd x^k$
and $h_k=\mrd^k h(x_0)/\mrd x^k$.  
From the chain rule we can easily find that
\begin{align}\label{eq:chain4}
h_4 = f_4g_1^4 + 6f_3g_1^2g_2 + 3f_2g_2^2+4f_2g_1g_3 + f_1g_4.
\end{align}
The derivative $f_k$  appears in $h_n$ in as many terms as there are distinct ways
of finding $k$ positive integers that sum to $n$. That number of terms
is known as the Stirling number of the second kind
\citep{grahamconcrete}. These Stirling numbers sum to the $n$'th Bell number
which grows rapidly with $n$.
We omit the  $m=d=1$ Faa di Bruno formula for arbitrary $n$ and present
instead the generalization due to \cite{cons:savi:1996}.

In the multivariate setting, $h(\bsx) = f(\bsg(\bsx))$
where $\bsx\in\real^m$ and $\bsy = \bsg(\bsx)\in\real^d.$
In our applications $\bsx\in[0,1]^m$. 
We write 
$\bsg(\bsx) = (g^{(1)}(\bsx),\dots,g^{(d)}(\bsx))$. 
The multivariate Faa di Bruno formula gives 
an arbitrary mixed partial derivative of $h$
with respect to components of $\bsx$ in terms 
of partial derivatives of $f$ and $g^{(i)}$. 
The formula requires that the needed derivatives exist. 

The formula uses some multi-index notation. 
We use $\natu_0$ for the set of non-negative integers.
Let $\bsnu =(\nu_1,\dots,\nu_m)\in\natu_0^m$.
Then $h_\bsnu$ is the derivative of $h$ taken $\nu_i$
times with respect to $x_i$. Similarly $f$ and $g^{(i)}$ subscripted 
by tuples of $d$ and $m$ nonnegative integers respectively, are 
the corresponding partial derivatives.  When the subscript is 
all zeros, the result is the function itself, undifferentiated. 

For a multi-index $\bsnu\in\natu_0^m$ we write 
$|\bsnu|=\sum_{i=1}^m\nu_i$ and 
$\bsnu! = \prod_{i=1}^m(\nu_i!)$. 
For $\bsz\in\real^m$ and $\bsnu\in\natu_0^m$ we write 
$\bsz^\bsnu$ for $\prod_{i=1}^m z_i^{\nu_i}$. 
We use an ordering $\prec$ on $\natu_0^m$ where 
$\bsmu \prec\bsnu$ means that either 
$|\bsmu|<|\bsnu|$, or $|\bsmu| = |\bsnu|$ holds 
along with $\mu_i<\nu_i$ at the smallest $i$ where 
$\mu_i\ne\nu_i$. Multi-indices in $\natu_0^d$ are treated the same way. The quantity $\bsg_{\bsell}$ is the vector 
$(g^{(1)}_\bsell,\dots,g^{(d)}_\bsell)$. 


\begin{theorem}\label{thm:fdb}
Let $g^{(i)}_\bsmu$ be continuous in a neighborhood of 
$\bsx_0\in\real^m$ for all $|\bsmu|\le|\bsnu|$ and 
all $i=1,\dots,d$, where $\bsmu,\bsnu\in\natu_0^m-\{\bszero\}$. 
Similarly, let $f_\bslam$ be continuous in a neighborhood 
of $\bsy_0 = \bsg(\bsx_0)$ for all $|\bslam|\le|\bsnu|$. 
Then in a neighborhood of $\bsx_0$,
\begin{align}\label{eq:consavfdb} 
h_\bsnu &= \bsnu!
\sum_{1\le|\bslam|\le|\bsnu|}f_\bslam \sum_{s=1}^{|\bsnu|}
\sum_{\calkl(s,\bsnu,\bslam)} \prod_{r=1}^s\frac{ [\bsg_{\bsell_r}]^{\bsk_r}}
{(\bsk_r!)[\bsell_r!]^{|\bsk_r|}},
\end{align} 
where
\begin{align}\label{eq:defkl}
\begin{split}
\calkl(s,\bsnu,\bslam) 
& = \Bigl\{ 
(\bsk_1,\dots,\bsk_s,
\bsell_1,\dots,\bsell_s)\in\bigl(\natu_0^d-\{\bszero\}\bigr)^s\times\bigl(\natu_0^m-\{\bszero\}\bigr)^s 
\mid\\
&\qquad \bsell_1\prec\bsell_2\prec\dots\prec\bsell_s,\
\sum_{r=1}^s\bsk_r=\bslam,\ \text{and}\ 
\sum_{r=1}^s|\bsk_r|\bsell_r=\bsnu 
\Bigr\}, 
\end{split}
\end{align}
\end{theorem}
\begin{proof}
\citet[Theorem 1]{cons:savi:1996}.
\end{proof}

For our purposes of comparing geometric
nets to (R)QMC, we only need $h_\bsnu$
with $\bsnu\in\{0,1\}^m-\{\bszero\}$.
Equations~\eqref{eq:consavfdb} and~\eqref{eq:defkl} simplify.
\begin{lemma}\label{lem:simplifyfdb}
For $m\ge1$ let $\bsnu\in\{0,1\}^m-\{\bszero\}$.
If $1\le|\bslam|\le|\bsnu|$ and $1\le s\le |\bsnu|$
and $(\bsk_1,\dots,\bsk_s,\bsell_1,\dots,\bsell_s)\in 
\calkl=\calkl(s,\bsnu,\bslam)$, then  for $r=1,\dots,s$,
$\bsk_r\in\{0,1\}^d-\{\bszero\}$,
$\bsell_r\in\{0,1\}^m-\{\bszero\}$ and $|\bsk_r|=1$.
Also $\bsnu!=\bsk_r!=\bsell_r!=1$.
\end{lemma}
\begin{proof}
Definition~\eqref{eq:defkl} of $\calkl$ includes the  
condition $\sum_{r=1}^s|\bsk_r|\bsell_r=\bsnu$.  
Because $\bsnu$ is a binary vector and $|\bsk_r|\ge1$, 
no component of $\bsell_r$ can be larger than $1$.  
Therefore $\bsell_r\in\{0,1\}^m-\{\bszero\}$.  
Similarly $\bsell_r$ has at least one nonzero component  
and so $|\bsk_r|\le1$. Because $\bsk_r\ne\bszero$ we now have  
$|\bsk_r|=1$.  Finally, the factorial of any binary vector is $1$.
\end{proof}

It follows from Lemma~\ref{lem:simplifyfdb}
that for $\bsnu\in\{0,1\}^m-\{\bszero\}$,
\begin{align}\label{eq:consavfdbbinary} 
h_\bsnu &= 
\sum_{1\le|\bslam|\le|\bsnu|} f_\bslam \sum_{s=1}^{|\bsnu|}
\sum_{\calkl(s,\bsnu,\bslam)} \prod_{r=1}^s
[\bsg_{\bsell_r}]^{\bsk_r}.  
\end{align} 
Next, we use Lemma~\ref{lem:simplifyfdb} to simplify
the derivatives of $g$.
Because $\bsnu$ is a nonzero binary vector, we can identify it
with a nonempty subset $v$ of $1{:}m$.
Specifically, $j\in v$ if and only if $\nu_j=1$.
Similarly we may identify the binary vector 
$\bsell_r\in\{0,1\}^m-\{\bszero\}$
with the set $\ell_r\subseteq1{:}m$.
The nonzero binary vector $\bsk_r\in\{0,1\}^d$ corresponds
to a singleton set.  We can therefore identify it with
an integer in $1{:}d$. We identify $\bsk_r$ with the integer $k_r$
such that $\bsk_{ri}=1$ if and only if $i=k_r$.

With this identification, 
\begin{align}\label{eq:thegprods}
[\bsg_{\bsell_r}]^{\bsk_r} 
= 
\prod_{i=1}^d\biggl(\frac{\partial^{|\bsell_r|} g^{(i)}}{
\prod_{r=1}^s
\partial x_r^{\ell_{jr}}}\biggr)^{k_{ri}}
= 
\prod_{i=1}^d\bigl(
\partial^{\ell_r} g^{(i)}\bigr)^{k_{ri}}
= 
\partial^{\ell_r} g^{(k_r)}.
\end{align}
Now switching from $g^{(k)}$ back to $\tau_k$ we get a Faa di Bruno
formula for mixed partial derivatives taken at most once with respect
to every index:
\begin{align}\label{eq:consavfdbbinaryreduced}
\partial^v (f\circ \tau) &= 
\sum_{\bslam\in\natu_0^m\atop
1\le|\bslam|\le|v|} f_\bslam \sum_{s=1}^{|v|}
\sum_{(\ell_r,k_r)\in {\calklt}(s,v,\bslam)} 
\prod_{r=1}^s
\partial^{\ell_r} \tau_{k_r}
\end{align}
where $\calklt(s,v,\bslam)$ equals
\begin{align}\label{eq:defklt}
\begin{split}
&\Bigl\{ (\ell_r,k_r),\ r=1,\dots,s,\ \Bigm|
\ell_r\subseteq1{:}m,\ 
k_r\in1{:}d,\ \cup_{r=1}^s\ell_r=v,\\
&\qquad 
\ell_r\cap\ell_{r'} =\emptyset
\ \text{for}\ r\ne r'\ 
\text{and}\ |\{j\in1{:}s\mid k_j=i\}|=\lambda_i
\Bigr\}. 
\end{split}
\end{align}

\section{Necessary and sufficient conditions}\label{sec:conditions}

Equation~\eqref{eq:consavfdbbinaryreduced}
allows us to find sufficient conditions on $f$ and $\tau$
so that $f\circ\tau\in\bvhk$. 
Uniformity transformations $\tau$ satisfy~\eqref{eq:tauok}.
We don't need that condition for $f\circ\tau\in\bvhk$ but we do
need it to ensure that averages of $f\circ\tau$ over $[0,1]^m$
estimate $\vol(\px)^{-1}\int_{\px}f(\bsx)\rd\bsx$.
We find conditions on $\tau$ so that $f\circ\tau\in\bvhk$ 
for all $f\in C^m(\px)$. Similarly we find 
conditions under which $f\circ \tau$ is smooth in the sense of
equation~\eqref{eq:snetcond} for all $f\in C^m(\px)$.

We will use a generalized H\"older inequality
\cite[page 141]{boga:2007}.  
For a positive integer $s$
suppose that $f_r\in L^{p_r}(\mu)$ for $r=1,\dots,s$, for some nonnegative measure $\mu$,
and that $\sum_{r=1}^s1/p_r = 1/p$. Then 
$
(\int |\prod_{r=1}^sf_r|^p\rd\mu)^{1/p}
\le \prod_{r=1}^s(\int |f_r|^{p_r}\rd\mu)^{1/p_r}
$
and so  $\prod_{r=1}^sf_r\in L^p(\mu)$.

\begin{theorem}\label{thm:bvhk}
Let $\tau(\bsu) = (\tau_1(\bsu) , \ldots, \tau_d(\bsu))$ be a map from $[0,1]^m$ to 
the closed and bounded set $\px\subset\real^d$ such 
that $\partial^{1:m}\tau_j$ exists for all $j=1,\dots,d$. 
Assume that 
$\partial^v \tau_j(\bsu_v{:}\bsone_{-v}) \in L^{p_j}\bigl([0,1]^{|v|}\bigr)$ 
for all $j = 1, \ldots, d$ and for all nonempty $v \subseteq 1{:}m$,
where $p_j\in[1,\infty]$. If $\sum_{j=1}^d1/p_j\le1$
then $f\circ\tau\in\bvhk$ for all $f\in C^d(\px)$. 
\end{theorem}
\begin{proof}
From~\eqref{eq:hkbound} we need to show that
$\partial^v(f\circ\tau) (\bsu_v{:}\bsone_{-v})\in L_1[0,1]^{|v|}$
for nonempty $v\subseteq1{:}m$.
From~\eqref{eq:consavfdbbinaryreduced}
it suffices to show that
\begin{align}\label{eq:toshowforsuff}
\int_{[0,1]^{|v|}} \Bigl|f_{\bslam}(\tau(\bsu_v{:}\bsone_{-v}))
\prod_{r=1}^s\partial^{\ell_r}\tau_{k_r}(\bsx_u{:}\bsone_{-u})\Bigr|\rd\bsx_u <\infty
\end{align}
for all $\bslam\in\natu_0^d$ with $1\le|\bslam|\le |v|$, 
all $s\le|v|$, all disjoint $\ell_r\subseteq v$ with union $v$, where
$\lambda_i$ of the $k_r$ are equal to $i$ for  $i\in1{:}d$.
Because $f_\bslam$ is uniformly bounded the integral
in~\eqref{eq:toshowforsuff} can be bounded by a product of univariable
integrals. The result 
then follows from our assumptions and from the generalized
H\"older inequality.
\end{proof}

As a special case,  consider $p_j=d$ for $j=1,\dots,d$. Notice that
the moment conditions on $\tau$ become more stringent as the dimension $d$ of
$\px$ increases.   
Attaining the usual QMC rate becomes increasingly difficult in higher
dimensions for QMC applied through a transformation $\tau$.

Next we consider the kind of smoothness that
allows scrambled nets to improve upon the quasi-Monte Carlo rate.
In this setting we require mixed partial derivatives in $L^2$ but
we do not have to pay special attention to components of $\bsu$
that equal $1$.

\begin{theorem}\label{thm:rqmcsmooth}
Let $\tau(\bsu) = (\tau_1(\bsu) , \ldots, \tau_d(\bsu))$ be a map from $[0,1]^m$ to 
the closed and bounded set $\px\subset\real^d$ such 
that $\partial^{1:m}\tau_j$ exists for all $j=1,\dots,d$. 
Assume that 
$\partial^v \tau_j\in L^{p_j}([0,1]^{m})$ 
for all $j = 1, \ldots, d$ and for all nonempty $v\subseteq 1{:}m$,
where $p_j\in[1,\infty]$.
If $\sum_{j=1}^d1/p_j\le1/2$,
then $f\circ\tau$ is smooth in the sense of
equation~\eqref{eq:snetcond} for all $f\in C^d(\px)$.
\end{theorem}
\begin{proof}
Essentially the same argument that proves Theorem~\ref{thm:bvhk}
applies here.
\end{proof}

Necessary conditions are more subtle.  To take an extreme 
case, $\tau$ could fail to be in $\bvhk$ or to be smooth, but $f$ could repair 
that problem by being constant everywhere, or just in a 
region outside of which $\tau$ is well behaved.  
Our working definition is that we consider a transformation $\tau$
to be unsuitable for QMC when one or more of the 
components $\tau_j$ has $\partial^u\tau_j(\cdot:1_{-v})\not\in L_1$
for some $v\subset1{:}m$.  In that case even
the coordinate function $\tau_j\not\in\bvhk$.
Similarly if $\partial^v\tau_j\not\in L^2$ for any $j$ and $v$ then
the transformation $\tau$ is one that does not lead to the improved
rate for scrambled nets even for integration of $\tau$, much
less $f\circ\tau$ for all $f\in C^d(\px)$.

It is possible to weaken the condition on $f$
in Theorem~\ref{thm:bvhk}, while strengthening 
the conditions on $\tau_j$. We could 
require only that $(f_\bslam\circ\tau)(\cdot{:}\bsone_{-v})$ is in $L^{p_0}([0,1]^{|v|})$ 
whenever $|\bslam|\le m$ and then require $\sum_{j=0}^d1/p_j\le 1$. 
Similarly for Theorem~\ref{thm:rqmcsmooth}, we could require
$f_\bslam\circ\tau\in L^{p_0}$ whenever $|\bslam|\le m$ where
$\sum_{j=0}^d1/p_j\le 1/2$.

\section{Counter-Examples}\label{sec:counter}
In this section we give two common transformations for which some $\tau_j \not\in \bvhk$, which means we do not satisfy the conditions of Theorem~\ref{thm:bvhk}. Thus unless we are very lucky, we would have $\vhk(f \circ \tau) = \infty$.

\subsection{Map from $[0,1]^3$ to an equilateral triangle}\label{sec:3to2}
Let $T^2 = \{ \bsx\in[0,1]^3\mid \sum_{j=1}^3x_j= 1\}$, an equilateral triangle. 
Consider the map $\tau : [0,1]^3 \rightarrow T^2$ defined by
\begin{equation}\label{eq:3to2}
\tau_j(\bsu) = \frac{\log u_j}{\sum_{i=1}^3 \log u_i},\quad j=1,2,3.
\end{equation}
It is well known that $\tau(\bsu)\sim\dustd(T^2)$ when $\bsu\sim\dustd([0,1]^3)$.
The mapping in~\eqref{eq:3to2} is well defined for $\bsu\in(0,1)^3$. There are various
reasonable ways to extend it to problematic boundary points with either some $u_j=0$ or
with all $u_j=1$.  We will show that none of those extensions can yield $\tau_j\in\bvhk$.

First we find that
\begin{align*}
\begin{split}
\iint_{(0,1)^2}\left|\frac{\partial^2 \tau_1}{\partial u_1 \partial u_2}\right|_{u_3 = 1} \rd u_1\rd u_2 
&= \iint_{(0,1)^2}\left|\frac{\log u_1 - \log u_2}{u_1 u_2 (\sum_{i=1}^2 \log u_i)^3}\right| \rd u_1\rd u_2.
\end{split}
\end{align*}
After a change of variable to
$x_1 = \log u_1$ and $x_2 = \log u_2$ the integral becomes
\begin{align*}
\begin{split}
& \int_{-\infty}^0\int_{-\infty}^0 \left|\frac{x_1 - x_2}{(x_1 + x_2)^3}\right|\rd x_1 \rd x_2 \\
&= \int_{-\infty}^0\int_{-\infty}^{x_1} \frac{x_1 - x_2}{(x_1 + x_2)^3} \rd x_2 \rd x_1 +  \int_{-\infty}^0\int_{x_1}^0 \frac{x_2 - x_1}{(x_1 + x_2)^3} \rd x_2 \rd x_1\\
&= \int_{-\infty}^0 \frac{1}{2x_1}dx_1 = \infty.
\end{split}
\end{align*}

Thus $\tau \not\in \bvhk$.
There is no extension from $(0,1)^3$ to $[0,1]^3$ that would yield $\tau\in\bvhk$.
The same argument applies if
$\tau_j(\bsu)=\log(u_j)/\sum_{i=1}^m\log(u_i)$ for any $m\ge 3$,
mapping $[0,1]^m$ to a $d=m-1$ dimensional simplex.  We can set
$u_i=1$ for $i\ge3$ and integrate as before.


\subsection{Inverse Gaussian map to the hypersphere}

A very convenient way to sample uniformly from 
the sphere $\mathbb{S}^{d-1}=\{\bsx\in\real^d\mid\Vert\bsx\Vert=1\}$
is to generate $d$ independent $\dnorm(0,1)$ random variables and standardize them.
We write $\varphi$ and $\Phi$ for the probability density function and cumulative
distribution function, respectively, of $\dnorm(0,1)$.
The mapping from $[0,1]^d$ to $\px=\mathbb{S}^{d-1}$ is
\[\tau_j(\bsu) = \frac{\Phi^{-1}(u_j)}{\sqrt{\sum_{i=1}^d \Phi^{-1}(u_i)^2}}.\]

We will use the double factorial function $n!!=n(n-2)(n-4)\cdots1$ for odd $n$
and set $g(n)=(2n-1)!!$.
For $j\in1{:}d$ and $v\subset1{:}d$ with $j\not\in v$ we find that
\begin{equation*}
\partial^v\tau_j = \frac{(-1)^{|v|} g(|v|)\Phi^{-1}(u_j)}
{\bigl(\sum_{i=1}^d \Phi^{-1}(u_i)^2 \bigr)^{|v|+1/2}}
\times 
 \prod_{i \in v} \frac{\Phi^{-1}(u_i)}{\varphi(\Phi^{-1}(u_i))}.
\end{equation*}
Now if $j \in v$, we can write after some algebra,
\begin{align}\nonumber 
\partial^v\tau_j &=  \frac{\partial}{\partial u_j} \partial^{v - \{j\}}\tau_j \nonumber \\
&= \frac{(-1)^{|v|} g(|v|) \prod_{i \in v - \{j\}} \Phi^{-1}(u_i) \bigl[2(|v| - 1)\Phi^{-1}(u_j)^2 - \sum_{i \neq j} \Phi^{-1}(u_i)^2\bigr] }{
\bigl(\sum_{i=1}^d \Phi^{-1}(u_i)^2 \bigr)^{|v|+1/2}
\prod_{i \in v} \phi(\Phi^{-1}(u_i))}.\nonumber
\end{align}

We choose $v=1{:}d$ and integrate $|\partial^{1{:}d}\tau_j|$ over $[0,1]^d$.
The integration is done with a change of variable $x_i = \Phi^{-1}(u_i)$
so $\rd u_i = \varphi(x_i)\rd x_i$. Because $g(d-1)\ge1$,
\[
\begin{split}
\int_{[0,1]^d} \bigl|\partial^{1:d}\tau_j\bigr| \rd u 
&\ge\int_{\real^d}  \Biggl|
\frac
{\prod_{i \neq j}x_i \bigl[2(d - 1)x_j^2 - \sum_{i \neq j} x_i^2\bigr] }
{\bigl(\sum_{i=1}^d x_i^2 \bigr)^{d+1/2}}\Biggr| \rd\bsx\\
&\ge  \int_{[0,\infty)^{d-1}} 
\int_0^\infty \frac{\left|\prod_{i \neq j}x_i \bigl[2(d - 1)x_j^2 - \sum_{i \neq j} x_i^2\bigr] \right|}{\bigl(\sum_{i=1}^d x_i^2 \bigr)^{d+1/2}} \rd x_j \rd \bsx_{-j}\\
&\ge 
\int_{[0,\infty)^{d-1}} 
\int_0^{\bigl(\frac{\sum_{i \neq j} x_i^2}{2(d-1)}\bigr)^{1/2}}\frac{\prod_{i \neq j}x_i \bigl[\sum_{i \neq j} x_i^2  - 2(d - 1)x_j^2\bigr]}{\bigl(\sum_{i=1}^d x_i^2 \bigr)^{d+1/2}} \rd x_j \rd\bsx_{-j}\\
&\ge \int_{[0,\infty)^{d-1}} \prod_{i \neq j}x_i \int_0^{\bigl(\frac{\sum_{i \neq j} x_i^2}{2(d-1)}\bigr)^{1/2}}\frac{ \bigl[\sum_{i \neq j} x_i^2  - 2(d - 1)x_j^2\bigr]}{\bigl(\frac{2d-1}{2d-2}\sum_{i\neq j}^d x_i^2 \bigr)^{d+1/2}} \rd x_j \rd\bsx_{-j}\\
&= \tilde{K}\int_{[0,\infty)^{d-1}} \frac{\prod_{i \neq j}x_i}{\bigl(\sum_{i \neq j} x_i^2 \bigr)^{d-1}} \rd\bsx_{-j}
\end{split}
\]
where $\tilde K=((2d-2)/(2d-1))^{d+1/2}$.

Now we integrate this one at a time for each $i \ne j$. Note that for $k < d-1$,
\begin{equation}
\label{eq_main1}
\int_0^\infty \frac{x}{(x^2 + z)^{d - k}} dx= c_k\frac{1}{z^{d-k-1}}
\end{equation}
where $c_k = 1/(2(d-k-1))$. Applying~\eqref{eq_main1} repeatedly for $k = 1$ to $k = d-2$, we get
\[
\begin{split}
\int_{[0,1]^d} \left|\partial^{1:d}\tau_j\right| \rd\bsu &\ge \left(\tilde{K}\prod_{k=1}^{d-2} c_k\right) \int_0^\infty \frac{1}{x_j}\rd x_j = \infty
\end{split}
\]
and so $\tau_j \not\in \bvhk$ for all $j\in1{:}d$.

\section{Mappings from \cite{fang1993number}}\label{sec:fangwang}

\cite{fang1993number} provide mappings from the unit 
cube to other important spaces for quadrature problems. 
Their mappings are more complicated than the elegant symmetric
ones in Section~\ref{sec:counter}.  
Instead of symmetry, their mappings are designed
to use a unit cube of exactly the same dimension
as the space they map too.
The domains that they consider, and their nomenclature for them,  are:
\begin{align}
\begin{split}
A_d &= \{(x_1, \ldots, x_d) : 0 \le x_1 \le \ldots \le x_d \le 1 \} \\
B_d &= \{(x_1, \ldots, x_d) : x_1^2 + \ldots + x_d^2 \le 1\} \\
U_d &= \{(x_1, \ldots, x_d) : x_1^2 + \ldots + x_d^2 = 1\} \\
V_d &= \{(x_1, \ldots, x_d)  \in \mathbb{R}_+^d : x_1 + \ldots + x_d \le 1\},\quad\text{and}\\
T_d &= \{(x_1, \ldots, x_d)  \in \mathbb{R}_+^d : x_1 + \ldots + x_d = 1\}. 
\end{split}
\end{align} 
Spaces $A_d$, $V_d$ and $T_{d+1}$ are all simplices of 
dimension $d$, $B_d$ is a ball and $U_d$ is the $d-1$
dimensional hyper-sphere.

We show next that all of their mappings have components $\tau$ in BVHK
and none of them have all mixed partial derivatives in $L^2$.
They are thus better suited to QMC than the symmetric mappings are
but they are not able to take advantage of the improved rate for RQMC
versus QMC.
The transformations have a separable character that lets 
us study them directly without recourse to the generalized H\"older inequality.


\subsection{Mapping from $[0,1]^d$ to $A_d$}
\label{sec:secAd}
The map $\tau = (\tau_1, \ldots, \tau_d)$ is given by $\tau_j(\bsu) = \prod_{i=j}^d u_i^{1/i}$ for $ j = 1, \ldots, d$. 
We find that
\begin{align*}
(\partial^{1:d}\tau_1)^2 = \prod_{i=1}^d \frac{1}{i^2} u_i^{2/i - 2}
\end{align*}
which diverges on integrating with respect to $u_2$. Thus $\partial^{1{:}d}\tau_1$  is not in $L^2$,
outside the trivial case $d=1$.

Next, we show that $\tau_j$ is in BVHK. 
Pick any non-empty $v\subseteq 1{:}d$. If there exists $i\in v$ with $i<j$, then
$\partial^{v}\tau_j = 0$, so we may assume that  $v\subseteq j{:}d$.
For such a $v$, we have,
\begin{align*}
\int_{[0,1]^{|v|}} |\partial^v \tau_j(\bsu_{v}{:}\bsone_{-u})|\rd\bsu_v = \int_{[0,1]^{|v|}} 
\biggl|\prod_{i \in v}\frac{1}{i}u_i^{1/i - 1}\biggr|
\rd\bsu_v  = 1. 
\end{align*}
Thus, using \eqref{eq:hkbound} we find that $\tau_j$ are BVHK for all $j$.

\subsection{Mapping from $[0,1]^d$ to $B_d$}
\label{sec:secBd}
The mapping involves the inverse transform of a distribution function on $B_d$. 
Define,
\[
F_j(x) = \begin{cases} x^d, &\mbox{if } j = 1 \\
\frac{\pi}{B\left(\frac{1}{2}, \frac{d-j+1}{2}\right)}\int_0^x (\sin \pi t)^{d-j} \rd t, 
& \mbox{if } j = 2, \ldots, d\end{cases}  
\]
where $B(\cdot,\cdot)$ is the Beta function. Next define intermediate variables
\begin{align*}
b_1 &= u_1^{1/d} \quad\text{and}\quad
b_i = F_i^{-1}(u_i), \qquad \text{for } i = 2, \ldots, d. 
\end{align*}
Their mappings are then
\begin{align*}
\tau_j &= b_1 \prod_{i=2}^j \sin(\pi b_i) \cos(\pi b_{j+1}), \quad \text{for } 1\le j\le d-2,\\
\tau_{d-1} &= b_1 \prod_{i=2}^{d-1}\sin(\pi b_i) \cos(2\pi b_d), \quad\text{and}\\
\tau_d &= b_1 \prod_{i=2}^{d-1}\sin(\pi b_i) \sin(2\pi b_d).
\end{align*}
For $d = 2$ we get $F_2(x) = x$ and so $\tau_2 = u_1^{1/2} \sin(2\pi u_2)$. Therefore
$\partial^{1:2} \tau_2 = {\pi \cos(2\pi u_2)}/{\sqrt{u_1}}$ which is not in $L^2$. For general $d > 2$, we have, 
\begin{align}\label{eq:forlater1}
\partial^{1:d} \tau_d = \frac{1}{d}u_1^{1/d - 1} \left(\prod_{i=2}^{d-1} \pi \cos(\pi b_i) \frac{\partial b_i}{\partial u_i}\right)  2\pi \cos(2\pi u_d), 
\end{align}
which is also not in $L^2$ because of the factor $u_1^{1/d-1}$.

For later use with the transformation to $U_d$, we also consider the factor for $i=d-1$ 
in~\eqref{eq:forlater1}.
The definition of $b_{d-1}$ simplifies to $b_{d-1} = \cos^{-1}(1 - 2 u_{d-1}) /\pi$ and so
 \begin{align*}
 \frac{\partial b_{d-1}}{\partial u_{d-1}} = \frac{2}{\pi \sin(\pi b_{d-1})}. 
 \end{align*}
 This simplifies the above mixed partial to
 \begin{align*}
 \partial^{1:d} \tau_d = \frac{1}{d}u_1^{1/d - 1} \left(\prod_{i=2}^{d-2} \pi \cos(\pi b_i) \frac{\partial b_i}{\partial u_i}\right) \left(2 \frac{(1-2u_{d-1})}{\sqrt{1 - (1-2u_{d-1})^2}}\right) 2\pi \cos(2\pi u_d) 
 \end{align*}
Now
 \begin{align} \label{int1}
 \int_0^1 \frac{(1-2u_{d-1})^2}{1 - (1-2u_{d-1})^2} du_{d-1} = \frac{1}{4} \left( \log x - \log (1-x) - 4x \right)\Bigg|_0^1 = \infty. 
 \end{align}

To show that this transformation is in BVHK we follow the proof of Theorem \ref{thm:bvhk}. Assuming $f \in C^d$, using \eqref{eq:toshowforsuff} it is enough to show that
\begin{align*}
\int_{[0,1]^{|v|}} \Bigl|
\prod_{r=1}^s\partial^{\ell_r}\tau_{k_r}(\bsu_v{:}\bsone_{-v})\Bigr|\rd\bsu_v <\infty,
\end{align*}
for $v\subseteq1{:}m$, $s\le |v|$ and $k_r\in1{:}d$.
Note that $\ell_r \cap \ell_{r'} = \emptyset$ for all $r \neq r'$ and $\cup_{r} \ell_r = v$. 
Thus, we differentiate at most once with respect to any original variable and we get,
\begin{align*}
\int_{[0,1]^{|v|}} \Bigl|
\prod_{r=1}^s\partial^{\ell_r}\tau_{k_r}(\bsu_v{:}\bsone_{-v})\Bigr|\rd\bsu_v \le \int_{[0,1]^{|v|}} \Bigl|
\prod_{i \in v } \frac{\partial b_i}{\partial u_i}\Bigr|\rd\bsu_v \le  \prod_{i \in v }  \int_{[0,1]} \Bigl|
\frac{\partial b_i}{\partial u_i}\Bigr|\rd u_i. 
\end{align*}
Now
\begin{align*}
\frac{\partial b_i}{\partial u_i} = \frac{B\left(\frac{1}{2}, \frac{d-j+1}{2}\right)}{\pi \left[\sin\left(\pi F_i^{-1}(u_i)\right)\right]^{d-i}}.
\end{align*}
Note that $F_i^{-1}(u_i) \in [0,1]$ for all $u_i \in [0,1]$. Thus we have,
\begin{align*}
\int_{[0,1]} \Bigl|
\frac{\partial b_i}{\partial u_i}\Bigr|\rd u_i = \int_{[0,1]} \frac{\partial b_i}{\partial u_i} \rd u_i = 1 
\end{align*}
which shows that $\tau$ is BVHK.

\subsection{Mapping from $[0,1]^{d-1}$ to $U_d$}
\label{sctud}
This mapping is similar to one in $B_d$, with the densities being different. Here we have,
\begin{align*}
f_j(u) = \frac{\pi}{B\left(\frac{1}{2} , \frac{d-j}{2}\right)} (\sin(\pi u))^{d-j-1}
\end{align*}
and $b_i = F_i^{-1}(u_i)$ for $i = 1, \ldots, d-1$. The explicit transformation can be written as 
\begin{align*}
\tau_j &= \prod_{i=1}^{j-1} \sin(\pi b_i) \cos(\pi b_j) \qquad \text{for } j = 1, \ldots, d-2 \\
\tau_{d-1} &= \prod_{i=1}^{d-2} \sin(\pi b_i) \cos(2\pi b_{d-1}),\quad\text{and}\\
\tau_d &= \prod_{i=1}^{d-2} \sin(\pi b_i) \sin(2\pi b_{d-1}).
\end{align*}
We first consider the case $d = 2$. It is easy to see that 
\begin{align*}
\tau_1 = \cos(2\pi u_1)\quad\text{and}\quad
\tau_2 = \sin(2\pi u_1).
\end{align*} 
Note that in this case, $\partial^v\tau_j \in L^2$ for each $v \subseteq \{1,2\}$ and $j = 1,2$. This is the only case with this property, but then the set $U_d$ is intrinsically
one dimensional. 
For $d \ge 3$, consider the $(d - 2)$-th term in the expansion of $\partial^{1:d-1} \tau_d$,
\begin{align*}
\partial^{1:d-1} \tau_d = \left( \prod_{i=1}^{d-3} \pi \cos(\pi bi)\frac{\partial b_i}{\partial u_i}\right)  \left(2 \frac{(1-2u_{d-2})}{\sqrt{1 - (1-2u_{d-2})^2}}\right) 2\pi \cos(2\pi u_{d-1}). 
\end{align*}
Using~\eqref{eq:forlater1} as in the previous case, 
this proves that $\partial^{1:d-1} \tau_d \not\in L^2$. Furthermore, following the same argument in Section \ref{sec:secBd}, we may show that this transformation is also in BVHK. 

\subsection{Mapping from $[0,1]^d$ to $V_d$}
Here we have,
\begin{align*}
\tau_i &= u_1^{1/d} \prod_{j=2}^i u_j^{\frac{1}{d-j+1}} \Bigl( 1 - u_{i+1}^{\frac{1}{d-i}}\Bigr) \qquad \text{for } i = 1, \ldots, d-1,\quad\text{and}\\
\tau_d &= u_1^{1/d} \prod_{j=2}^d u_j^{\frac{1}{d-j+1}}.
\end{align*}
Considering the mixed partial, $\partial^{1:d}\tau_d$ we have 
\begin{align*}
\partial^{1:d}\tau_d &= \frac{1}{d}u_1^{\frac{1}{d} - 1}\prod_{j=2}^d \frac{1}{d-j+1} u_j^{\frac{1}{d-j+1}-1}
= \frac{1}{d!} \frac{1}{u_1^{\frac{d-1}{d}}u_2^{\frac{d-2}{d-1}} \ldots u_{d-1}^{1/2} } \nonumber. 
\end{align*}
Observing the integral with respect to $u_{d-1}$ it is clear that $\partial^{1:d}\tau_d \not\in L^2$. Furthermore, following the same argument in Section \ref{sec:secAd}, 
we may show that this transformation is also in BVHK.

\subsection{Mapping from $[0,1]^{d-1}$ to $T_d$}
Similar to $V_d$, here we have,
\begin{align*}
\tau_i &= \prod_{j=1}^{i-1} u_j^{\frac{1}{d-j}} \Bigl( 1 - u_{i}^{\frac{1}{d-i}}\Bigr) \qquad \text{for } i = 1, \ldots, d-1 \nonumber \\
\tau_d &= \prod_{j=1}^{d-1} u_j^{\frac{1}{d-j}}.
\end{align*}
It is thus clear from 
\begin{align*}
\partial^{1:d-1}\tau_d = \frac{1}{(s-1)!} \frac{1}{u_1^{\frac{d-2}{d-1}}u_2^{\frac{d-3}{d-2}} \ldots u_{d-2}^{1/2} u_{d-1}}
\end{align*}
that $\partial^{1:d-1}\tau_d \not\in L^2$. The argument from 
Section \ref{sec:secAd} shows that this transformation is in BVHK. 

\subsection{Efficient mapping from $[0,1]^{d-1}$ to $U_d$}
\cite{fang1993number} 
gave another mapping to $U_d$ which avoids computing the
incomplete beta function that was used 
in Section \ref{sctud}. Once again the transformation fails to have
all partial derivatives in $L^2$. We deal with the case of $d$ being even and odd differently. 

\subsubsection*{Even case: $d = 2m$}
Here we have $(u_1, \ldots, u_{d-1}) \in [0,1]^{d-1}$. Define $g_m = 1$ and $g_0 = 0$. For 
$j$ from $m-1$ down to $1$, let $g_j = g_{j+1} u_j^{1/j}$.
Put $d_l = \sqrt{g_l - g_{l-1}}$. Now for $l = 1,\ldots, m$, define
\begin{align*}
\tau_{2l - 1} &= d_l \cos(2\pi u_{m + l - 1}) \quad\text{and}\quad
\tau_{2l} = d_l \sin(2\pi u_{m + l - 1}).
\end{align*}
It is easy to see that 
\begin{align*}
\tau_1 = d_1 \cos(2\pi u_m) = \prod_{j=1}^{m-1} u_j^{1/2j} \cos(2\pi u_m)
\end{align*}
and so
\begin{align*}
|\partial^{1:m}\tau_1| &= \Biggl|\prod_{j=1}^{m-1} \frac{1}{2j} u_j^{1/2j-1} 2\pi \sin(2\pi u_m) \Biggr|
= \Biggl| \frac{1}{2^{m-1} (m-1)!}\frac{2\pi \sin(2\pi u_m)}{u_1^{\frac{1}{2}} u_2^{\frac{3}{4}} \ldots u_{m-1}^{\frac{2m - 3}{2m - 2}}} 
\Biggr|.
\end{align*}
Integrating with respect to $u_1$ shows that $\partial^{1:m}\tau_1 \not\in L^2$. 

\subsubsection*{Odd case: $d = 2m+1$}
Again we begin with $(u_1, \ldots, u_{d-1}) \in [0,1]^{d-1}$. Define $g_m = 1$ and $g_0 = 0$. For $j = m - 1$ down to $j=1$, let 
$g_j = g_{j+1} u_j^{{2}/(2j+1)}$. 
As for the even case, put $d_l = \sqrt{g_l - g_{l-1}}$. Now for $l = 2, \ldots, m$, define
\begin{align*}
\tau_1 &= d_1(1 - 2u_m),\\
\tau_2 &= d_1\sqrt{u_m(1 - u_m)} \cos(2\pi u_{m+1}),\\
\tau_3 &= d_1\sqrt{u_m(1 - u_m)} \sin(2\pi u_{m+1}),\quad\text{and then}\\
\tau_{2l} &= d_l \cos(2\pi u_{2l}),\quad\text{and}\quad
\tau_{2l+1} = d_l \sin(2\pi u_{2l}).
\end{align*}
Simplifying $d_1$ we see that 
\begin{align*}
\tau_2 = u_1^{\frac{1}{3}} u_2^{\frac{1}{5}} \ldots u_{m-1}^{\frac{1}{2m-1}} \sqrt{u_m(1 - u_m)} \cos(2\pi u_{m+1}) 
\end{align*}
Thus
\begin{align*}
\left|\frac{\partial \tau_2}{\partial u_m}\right|_{u_j = 1, j \neq m} = \frac{1 - 2 u_m}{2 \sqrt{u_m(1 - u_m)}}
\end{align*}
Applying (\ref{int1}) we see that $\partial^{u_m}\tau_2\not\in L^2$.

All of the $\tau_j$ are in BVHK. This 
follows from the fact that each component of the transformation is a 
product of functions of only one of the original variables and hence it is completely separable.  

\section{Nonuniform transformations}\label{sec:nonunif}

Here we consider transformations that are not uniformity preserving.
Section~\ref{sec:is} considers importance sampling methods for
integrals with respect to a non-uniform measure on $\cx$.
Section~\ref{sec:rhm} gives conditions for sequential inversion
to yield an integrand in BVHK.
Section~\ref{sec:issimplex} shows that some importance 
sampling transformations lead to the $O(n^{-3/2+\epsilon})$
rate for RMSE on the simplex for a class of functions including 
polynomials.

Suppose that
$\mu = \int_{\cx} f(\bsx) p(\bsx)\rd\bsx$
for a nonuniform density $p$. Instead of sampling drawing $\bsx_i$
from the uniform distribution on $\cx$ and averaging $fp$, a Monte
Carlo strategy can be to sample $\bsx_i\sim p$ and average $f$,
or under conditions, sample $\bsx_i\sim q$ and average $fp/q$.
This latter is importance sampling.

\cite{aistleitner2014functions} show that if $f$ is a measurable function on 
$[0,1]^d$ which is BVHK and $P$ is a normalized Borel measure on $[0,1]^d$, then
for $\bsx_1, \ldots, \bsx_n$ in $[0,1]^d$,
\begin{align}
\label{eq:generalKH}
\left| \frac{1}{n} \sum_{i=1}^n f(\bsx_i) - \int_{[0,1]^d} f(\bsx) \rd P(\bsx) \right| \le \vhk(f)D_n^*(\bsx_1, \ldots, \bsx_n ; P). 
\end{align}
where 
\begin{align*}
D_n^*(\bsx_1, \ldots, \bsx_n ; P) = \sup_{A \in \ca^*} \biggl| \frac{1}{n} \sum_{i=1}^n \mathds{1}_A(\bsx_i) - P(A)\biggr|
\end{align*}
and $\ca^*$ is the class of all closed axis-parallel boxes contained in $[0,1]^d$. 
\cite{aistleitner2013low}, prove that for any measure $P$ and any $n$ there 
exists of points $\bsx_1, \ldots, \bsx_n$ such that in
\begin{align*}
D_n^*(\bsx_1, \ldots, \bsx_n ; P) \le c_d (\log n)^{(3d+1)/2}n^{-1}. 
\end{align*}
They do not however give an explicit construction. 
Instead of using \eqref{eq:generalKH} one might use the original Koksma-Hlawka inequality by using an appropriate non-measure preserving transformation $\tau$. 
Below we give a corollary to Theorem \ref{thm:bvhk} stating the sufficient conditions for getting the optimal bound when using importance sampling.

\subsection{Importance Sampling}\label{sec:is}
We suppose that the measure $P$ has a probability density $p$.
We then use a transformation $\tau$ on $[0,1]^m$ which yields
$\bsx=\tau(\bsu)$ with probability density function $q$ on $\cx$
when $\bsu\sim\dustd[0,1]^m$.
We estimate $\mu$ by
\begin{align*}
\hat{\mu}_q^n = \frac{1}{n} \sum_{i=1}^n \frac{f(\tau(\bsu_i))p(\tau(\bsu_i))}{q(\tau(\bsu_i))} = \frac{1}{n} \sum_{i=1}^n \left(\frac{fp}{q} \circ \tau \right)(\bsu_i). 
\end{align*}
If $q(\bsx)>0$ whenever $p(\bsx)>0$ (and if $\mu$ exists) 
then $\e( \hat\mu_q^n)=\mu$.
Thus, to apply the Koksma-Hlawka inequality we only need $(fp/q) \circ \tau\in\bvhk[0,1]^m$.
Following Theorem \ref{thm:bvhk} we can now state the sufficient conditions for the above to hold. 
\begin{corollary}
Let $p$ and $q$ denote densities on $\cx$ with $q(\bsx)>0$ whenever $p(\bsx)>0$. 
Let $\tau$ be a map from $[0,1]^m$ to $\cx$ for which $\bsu \sim \dustd[0,1]^m$ 
implies $\bsx=\tau(\bsu)\sim q$. 
Moreover, assume that $\tau$ satisfies the conditions of Theorem \ref{thm:bvhk} and 
that $fp/q \in C^d(\cx)$. 
Then, for a low-discrepancy point set $\bsu_1, \ldots, \bsu_n$ in $[0,1]^m$,
\begin{align*}
\biggl| \int_\cx f(\bsx)p(\bsx)\rd\bsx- \frac{1}{n} \sum_{i=1}^n \left(\frac{fp}{q} \circ \tau \right)(\bsu_i) \biggr| = O\Bigl(\frac{(\log n)^{d-1}}{n}\Bigr).
\end{align*}
\end{corollary}
\begin{proof}
The result follows from Theorem \ref{thm:bvhk} and the Koksma-Hlawka inequality. 
\end{proof}

There is a similar counterpart to Theorem~\ref{thm:rqmcsmooth}.
When $\tau$ satisfies the conditions there, $fp/q\in C^d$,  and $\bsu_1,\dots,\bsu_n$
are a digital net with a nested uniform or linear scramble, then the RMSE 
of $\hat\mu^n_q$ is $O(n^{-3/2+\epsilon})$.
In both cases it is clearly advantageous to have $p/q$ bounded above,
just as is often recommended for importance sampling in Monte
Carlo. See for instance \cite{gewe:1989}.

\subsection{Rosenblatt-Hlawka-M\"uck Transformation}\label{sec:rhm}
For $d=1$, a standard way to generate a non-uniform random variable
is to invert the CDF at a uniformly distributed point. The multivariate
version of this procedure can be used to sample from an arbitrary
distribution provided we can invert all the conditional distributions
necessary.  

Let $F$ be the target distribution. Further let $F_1$ be the marginal distribution of $X_1$ and 
for $j = 2, \ldots, d$, let $F_{j\mid 1{:}(j-1)}(\,\cdot  \mid \bsx_{1:(j-1)})$ be the conditional CDF of $X_j$ 
given $X_1, \ldots, X_{j-1}$. The 
transformation $\tau$ of $\bsu \in [0,1]^d$ 
is given by $\bsx=\tau(\bsu) \in \mathbb{R}^d$ where 
\begin{equation}
\label{rosen_full}
 x_1 = F_1^{-1}(u_1) 
\quad \text{and} \quad 
x_j = F_{j\mid1{:}(j-1)}^{-1}(u_j \mid \bsx_{1:(j-1)}),\quad\text{for $j \ge 2$}. 
\end{equation}
The inverse transformation, from $\bsx$ to $\bsu$, was studied
by \cite{rosenblatt1952remarks}. 
\cite{hlaw:muck:1972}  
studied the use of this transformation for generating 
low discrepancy points.  
Under conditions on $F$ they show that the resulting points 
have a discrepancy with respect to $F$ of order $D_n^{1/d}$ where $D_n$ is the 
discrepancy of points $\bsu_1,\dots,\bsu_n$ that it uses. Because discrepancy 
can at best be $O(n^{-1}\log(n)^{d-1})$, that rate has a severe deterioration
with respect to dimension~$d$.
We suspect that this sequential inversion method was used before 1972
but have not found a reference.

We consider the case of $d=2$ dimensions.
Then $\tau(\bsu) = (\tau_1(\bsu), \tau_2(\bsu))$ where 
\begin{align}
\label{rosen}
\tau_1(u_1,u_2) &= F_1^{-1}(u_1)\quad\text{and}\quad
\tau_2(u_1,u_2) = F_{2|1}^{-1}(u_2 \mid F_1^{-1}(u_1)).
\end{align}

Let $f_{1,2}$ be the joint density, $f_1$ be the marginal density of $X_1$ corresponding to $F_1$ and $f_{2|1}(\cdot | x_1)$ be the conditional density of $X_2$ given $X_1 = x_1$. Then
\[
\begin{split}
\frac{\partial \tau_1}{\partial u_1} &= \frac{1}{f_1(F_1^{-1}(u_1))}
= \frac{1}{f_1(\tau_1(\bsu))},\quad\text{and} \\
\frac{\partial \tau_2}{\partial u_2} &= \frac{1}{f_{2|1}(F_{2\mid1}^{-1}(u_2 \mid F_1^{-1}(u_1)) | F_1^{-1}(u_1))}\\
&= \frac{f_1(F_1^{-1}(u_1))}{f_{1,2}(F_1^{-1}(u_1) , F_{2\mid1}^{-1}(u_2 \mid F_1^{-1}(u_1)))}
= \frac{f_1(\tau_1(\bsu))}{f_{1,2}(\tau(\bsu))}.
\end{split}
\]

It is easy to see that $\tau_1 \in\bvhk$ if the support of $X_1$ is finite, i.e., $F_1^{-1}(1) - F_1^{-1}(0) < \infty$. Thus, from here on, we assume that $X_1$ is defined on a compact set $[a,b]$, i.e., $a = F_1^{-1}(0)$ and $b = F_1^{-1}(1)$. 
We also use $a(x_1) = a_2(x_1) = F_{2\mid1}^{-1}(0\mid x_1)$
and $b(x_1) = b_2(x_1) = F_{2\mid1}^{-1}(1\mid x_1)$ which must both be finite if $\tau_2\in\bvhk$.

We now consider sufficient conditions for $\tau_2$ to be in BVHK. Assuming that $f_{1,2}$ is strictly positive, we have,
\begin{align*}
\int_{0}^{1} \left|\frac{\partial \tau_2}{\partial u_2}\right|_{u_1 = 1}   \rd u_2 
&= \int_0^1  \frac{1}{f_{2|1}(F_{2\mid1}^{-1}(u_2|b) | b)} \rd u_2\\ 
&= F_{2\mid1}^{-1}(1|b) - F_{2\mid1}^{-1}(0|b) = b_2(b)-a_2(b).
\end{align*}
Now define $F_{2|1}(x_1,x_2) = P(X_2 \le x_2 | X_1 = x_1)$, that is the conditional distribution of $X_2$ given $X_1 = x_1$. Further, let $f^{(k)}$ denote the partial derivative $\partial^{\{k\}}f$ for various functions $f$.
Using this notation, and the implicit function theorem, we obtain
\begin{align*}
\frac{\partial \tau_2(u_1, u_2)}{\partial u_1} 
&= \frac{\frac{ - \partial F_{2|1}}{\partial x_1} \Big|_{(F_1^{-1}(u_1), \tau_2(u_1,u_2))} \frac{1}{f_1(F_1^{-1}(u_1))}}{\frac{ \partial F_{2|1}}{\partial x_2} \Big|_{(F_1^{-1}(u_1), \tau_2(u_1,u_2))}}
= \frac{\frac{ - \partial F_{2|1}}{\partial x_1} \Big|_{\tau(\bsu)} 
}
{f_1(\tau_1(\bsu))\times\frac{ \partial F_{2|1}}{\partial x_2} \Big|_{ \tau(\bsu) }},
\end{align*}
where
\begin{align}
\label{eq:partial1}
\frac{ \partial F_{2|1}(x_1,x_2)}{\partial x_1} &= \int_{F_{2\mid1}^{-1}(1\mid x_1)}^{x_2} \frac{f_{1,2}^{(1)}(x_1,t) f_1(x_1) - f_{1,2}(x_1,t) f_1^{(1)}(x_1)}{f_1(x_1)^2} \rd t
\end{align}
and
\begin{align}
\label{eq:partial2}\frac{\partial F_{2|1}(x_1,x_2)}{\partial x_2} &= f_{2|1}(x_2|x_1).
\end{align}
Now we get, using a change of variable,
\begin{align*}
\int_{0}^{1} \left|\frac{\partial \tau_2}{\partial u_1}\right|_{u_2 = 1}   \rd u_1 \le 
\int_{a}^b\biggl|
{\frac{\partial F_{2|1}}{\partial x_1}}\Bigm/ {\frac{ \partial F_{2|1}}{\partial x_2}} 
\biggr|_{(x_1, b(x_1)}\rd x_1. 
\end{align*}

Finally to evaluate the complete mixed partial, differentiating $\tau_2$ with respect to $u_1$ and $u_2$ we have,
\[
\frac{\partial^2\tau_2}{\partial u_1 \partial u_2} = \frac{1}{f_{1,2}^2} \biggl( f_{1,2} \frac{f_1^{(1)}}{f_1} - f_{1,2}^{(1)} - f_{1,2}^{(2)}f_1 \frac{\partial \tau_2}{\partial u_1}\biggr).
\]
This gives us,
\[
\int_0^1 \int_0^1\left|\frac{\partial^2\tau_2}{\partial u_1 \partial u_2}\right|\rd u_1\rd u_2 = \int_0^1 \int_0^1 \frac{1}{f_{1,2}^2} \left| f_{1,2} \frac{f_1^{(1)}}{f_1} - f_{1,2}^{(1)} - f_{1,2}^{(2)}f_1 \frac{\partial \tau_2}{\partial u_1}\right| \rd u_1 \rd u_2. 
\]
Again by a change of variables via $x = F_1^{-1}(u_1)$ and $y = F_{2\mid1}^{-1}(u_2 \mid F_1^{-1}(u_1))$ we have,
\begin{align*}
& \int_0^1 \int_0^1 \left|\frac{\partial^2\tau_2}{\partial u_1 \partial u_2}\right|\rd u_1 \rd u_2 \\
&=  \int_a^b \int_{a(x)}^{b(x)}\left| \frac{f_1^{(1)}(x)}{f_1(x)} - \frac{1}{f_{1,2}(x,y)}\left(f_{1,2}^{(1)}(x,y) + f_{1,2}^{(2)}(x,y) \frac{\rd y}{\rd x}\right)\right| \rd y \rd x. 
\end{align*}
Let us write $\Delta$ for  the total derivative operator:
\begin{align*}
\Delta f &= \sum_{k=1}^d f^{(k)}(x_1,\ldots, x_d) \rd x_k,\quad\text{and}\\
\frac{\Delta f}{\rd x_j} &=  \sum_{k=1}^d f^{(k)}(x_1,\ldots, x_d) \frac{\rd x_k}{\rd x_j}.
\end{align*}
This allows us to write the integral above as 
\begin{align*}
 \int_0^1 \int_0^1 \left|\frac{\partial^2\tau_2}{\partial u_1 \partial u_2}\right|\rd u_1\rd u_2 
&=  \int_a^b \int_{a(x)}^{b(x)} \left| \frac{\Delta \log (f_1)}{\rd x} - \frac{\Delta \log (f_{1,2})}{\rd x}\right|\rd y \rd x. 
\end{align*}

Combining these, we now have a sufficient condition for the Rosenblatt transformation in two dimensions to be in $\bvhk$ which we summarize in Lemma~\ref{d2}. 

\begin{lemma}\label{lem:rose}
\label{d2}
Let $X_1$ be supported on the finite interval $[a,b]$ and,
given $X_1=x_1$, let $X_2$ be supported on the finite
interval $[a(x_1),b(x_1)]$.
Let $f_1$ and $f_{1,2}$ be the densities of $X_1$ and $(X_1, X_2)$ respectively.
Then the Rosenblatt-Hlawka-M\"uck 
transformation \eqref{rosen} is of bounded variation in the sense of Hardy and Krause if for each $k = 1,2$,
\begin{align}\label{eq:rose1}
 \int_{a}^b\int_{a(x_1)}^{b(x_1)} \left| \frac{\Delta \log (f_{1,\cdots,k})}{\rd x_1}\right|\rd x_2 \rd x_1 & < \infty,
\quad{\text{and}} \\
\int_{a}^b\,\biggl|
{\frac{\partial F_{2|1}}{\partial x_1}}\Bigm/ {\frac{ \partial F_{2|1}}{\partial x_2}} 
 \biggr|_{(x_1, b(x_1)}\rd x_1 &<\infty \label{eq:rose2}
\end{align}
where ${\partial F_{2|1}}/{\partial x_1}$ and ${\partial F_{2|1}}/{\partial x_2}$ are given by \eqref{eq:partial1} and \eqref{eq:partial2} respectively.

\end{lemma}

\medskip
From condition~\eqref{eq:rose1} we see
that the densities $f_1$ and $f_{1,2}$ can be problematic
as they approach zero on $\cx$, for then $\Delta \log f$ may become
large. Thus we anticipate better results when these densities are uniformly
bounded away from zero on $\cx$.
Condition~\eqref{eq:rose2} involves an integral over the upper
boundary of $\cx$.  If that upper boundary is flat, that is $b(x_1)$ is
constant on $a\le x_1\le b$, then the partial derivative in the numerator
there vanishes.
It is possible to generalize Lemma~\ref{lem:rose} to $d>2$ but
the resulting quantities become very difficult to interpret.

\subsection{Importance sampling QMC for the simplex}\label{sec:issimplex}

We map $\bsu=(u_1,u_2,\dots,u_d)\in[0,1]^d$ to 
$\bsx=\tau(\bsu)$ in the simplex 
$$A_d=\{(x_1,\dots,x_d)\in[0,1]^d\mid x_1 \le x_2\le\cdots\le x_d\}.$$
The mapping is given by 
$$
x_j = \tau_j(\bsu) = \prod_{k\ge j}u_k^{a_k}
$$
for constants $a_{k}>0$. 
The uniformity preserving mapping from \cite{fang1993number}  has $a_{k} = 1/k$. 

The Jacobian matrix for this transformation is upper triangular and hence 
the Jacobian determinant is 
$$
J(\bsu) = 
\prod_{j=1}^d \frac{\partial x_j}{\partial u_j}
= \prod_{j=1}^d a_{j}u_j^{a_j-1} \prod_{k>j}u_k^{a_k}
= A \prod_{j=1}^d u_j^{ja_j-1}
$$
where $A=\prod_ja_{j}$. 
The average of $J(\bsu)$ is $1/\vol(A_d) = 1/d!$ and $0\le J(\bsu)\le A$.  
The choice of \cite{fang1993number}
gives $J=1/d!$ for all $\bsu$. 
It is desirable to have $J$ be nearly constant.  
If $A\gg 1/d!$ then $J(\bsu)$ is a very `spiky' function and 
that will tend to defeat the purpose of low discrepancy sampling. 

The RQMC estimate of 
\begin{align*}
\mu &= d!\int_{A_d}f(\bsx)\rd\bsx = d!\int_{[0,1]^d}f(\tau(\bsu)) J(\bsu) \rd\bsu\quad\text{is}\\
\hat\mu &= \frac{d!}n\sum_{i=1}^n f(\tau(\bsu_i))J(\bsu_i). 
\end{align*}

Suppose that $f\in C^d$. 
Ignoring the $d!$ factor, 
the integrand on $[0,1]^d$ is now $\wt f(\bsu)=f(\tau(\bsu))J(\bsu)$, and 
$\partial^v \wt f 
= \sum_{w\subseteq v} \partial^w (f\circ\tau)\times \partial^{v-w}J. 
$
The definition of $\tau_j$ in this case makes it convenient to work with a simple function class 
consisting of integrands of the form $\prod_{j=1}^dx_j^{q_j}$ for real values $q_j\ge0$. 

\begin{theorem}\label{thm:rateforsimplex}
For $\bsx\in A_d$ let $f(\bsx) = \prod_{j=1}^d x_j^{q_j}$ 
for $q_j\ge0$. 
For $j=1,\dots,d$ and $\bsu\in[0,1]^d$, define $x_j=\tau_j(\bsu)=\prod_{k\ge j}u_k^{a_k}$
and the Jacobian $J(\bsu) = \prod_{j=1}^d a_ju_j^{ja_j-1}$ where $a_j>0$. 
Then $\partial^v f(\bsx(\bsu))J(\bsu)\in L^2[0,1]^d$ for all $v\subseteq1{:}d$
and all $q_j$ if and only if $a_j>3/(2j)$ holds for $j=1,\dots,d$. 
\end{theorem}
\begin{proof}
Let  $Q_k=\sum_{j\le k}q_k$ and  $A=\prod_ja_j$. Then 
$$\wt f(\bsu)=f(\boldsymbol{\tau}(\bsu))J(\bsu) 
=A\prod_{k=1}^d u_k^{ka_k-1+a_kQ_k}
$$
For $v\subseteq1{:}d$ we find that $(\partial^v\wt f(\bsu))^2$ equals 
\begin{align}\label{eq:l2simpintegrand}
A^2\prod_{k\in v}
({ka_k-1+a_kQ_k} )^2 u_k^{2(ka_k-2+a_kQ_k)}
\prod_{k\in -v}u_k^{2(ka_k-1+a_kQ_k)}. 
\end{align}
The coefficient ${ka_k-1+a_kQ_k}$ cannot vanish for all $q$. Therefore~\eqref{eq:l2simpintegrand}
has a finite integral for all $q_j$ if and only if 
$2(ja_j-2+a_jQ_j)>-1$ for all $j$ and all $q_1,\dots,q_d$. 
This easily holds if all $a_j>3/(2j)$.  
Conversely, suppose that $a_j\le 3/(2j)$
for some $j$. We may choose $Q_j=0$ and $v=\{j\}$ and see that~\eqref{eq:l2simpintegrand}
is not integrable. 
\end{proof}

From Theorem~\ref{thm:rateforsimplex} we see that RQMC can attain the 
$O(n^{-3/2+\epsilon})$ rate for 
functions of the form $\prod_{j=1}^d x_j^{q_j}$ on the simplex $A_d$. That 
rate extends to linear combinations of finitely many such functions, including 
polynomials and more. 
If we choose $a_j = 3/(2j)+\eta$ for some small $\eta>0$ then 
for fixed $d$ we have $d!J(\bsu) = (3/2)^d +O(\eta)$. 
There is thus a dimension effect. The integrand becomes more spiky 
as $d$ increases.  We can expect that the lead constant in the error 
bound will grow exponentially with $d$. 

For $d=1$, Theorem~\ref{thm:rateforsimplex}  requires $a_1>3/2$
whereas ordinary RQMC attains the $O(n^{-3/2+\epsilon})$
RMSE  with $a_1=1$ in that case. 
The reason for the difference is that the theorem covers more challenging 
integrands like $x_1^{1/2}$ whose derivative is not in $L^2$. 
If we work only with polynomials taking only $q_j\in\natu$, 
then the choice $a_k=1/k$ zeros out~\eqref{eq:l2simpintegrand}
when $Q_k=0$. The smallest nonzero $Q_k$ is then $1$ and we 
would need to impose $2(ka_k-2+a_kQ_k)>-1$. 
That simplifies to $Q_k>k/2$ which can only be ensured for $k=1$
and hence the Fang and Wang choice $a_k=1/k$ will not attain 
the RQMC rate for polynomials when $d\ge2$. 

We can extend  Theorem~\ref{thm:rateforsimplex} to all 
$f\in C^d$ via Theorem~\ref{thm:rqmcsmooth}, but only
for relatively large $a_j$. We require such large $a_j$
because the generalized H\"older inequality is conservative in 
this setting.

\begin{theorem}\label{thm:rateforsimplex2}
Let $f\in C^d(A_d)$, and define $x_j=\tau_j(\bsu) = \prod_{k=j}^du_j^{a_j}$ for 
$a_j>0$. Let $\tilde f(\bsu) = f(\tau(\bsu))J(\bsu)$ for the Jacobian 
$\prod_{j=1}^da_ju_j^{ja_j-1}$. If $a_1>3/2$ and  $a_j\ge 1$ for $2\le j\le d$, 
then $\partial^v\tilde f\in L^2[0,1]^d$. 
\end{theorem}
\begin{proof}
Define $\cx = A_d\times[0,A]\subset\real^{d+1}$ where $A=\prod_{j=1}^da_j$
and $\tau_{d+1}(\bsu) = J(\bsu)$. 
Then $\tilde f(\tau_1(\bsu),\dots,\tau_{d+1}(\bsu)) = 
f(\tau_1(\bsu),\dots,\tau_d(\bsu)) \tau_{d+1}(\bsu)\in C^{d+1}(\cx)$. 

For $j\le d$, and $v\subseteq1{:}d$ we have 
$\partial^v\tau_j =0$ unless $v\subseteq j{:}d$
and if  $v\subseteq j{:}d$, then 
$\partial^v\tau_j = \prod_{\ell\in v}a_\ell u_\ell^{a_\ell-1}\times\prod_{\ell\in j{:}d-v}u_\ell^{a_\ell}$.
Under the conditions of this theorem every
$\tau_j\in L^\infty[0,1]^d$. Next we can directly find that under the given conditions 
$\tau_{d+1} = J\in L^2[0,1]^d$. 
Then we have $\partial^v \tilde f\in L^2$
by Theorem~\ref{thm:rqmcsmooth}. 
\end{proof}

In Monte Carlo sampling, the effect of nonuniform importance 
sampling is sometimes measured via an effective sample 
size. See \cite{kong1994sequential}.
For the Jacobian above the effective sample 
size is the nominal one multiplied by 
$(\int J(\bsu)\rd\bsu)^2/\int J(\bsu)^2\rd\bsu$. If we take $a_j=3/(2j)$
this factor becomes $(8/9)^d$ which corresponds to a mild exponential 
decay in effectiveness for Monte Carlo sampling.  There seems to be
as yet no good measure of effective sample size for randomized QMC.



\section*{Acknowledgments}

This work was supported by the NSF under grant DMS-1407397.

\bibliographystyle{apalike}
\bibliography{qmc}
\end{document}